\documentclass[sigconf,nonacm]{aamas} 

\usepackage{balance} 

\setcopyright{ifaamas}
\acmConference[AAMAS '21]{Proc.\@ of the 20th International Conference on Autonomous Agents and Multiagent Systems (AAMAS 2021)}{May 3--7, 2021}{London, UK}{U.~Endriss, A.~Now\'{e}, F.~Dignum, A.~Lomuscio (eds.)}
\copyrightyear{2021}
\acmYear{2021}
\acmDOI{}
\acmPrice{}
\acmISBN{}

\title[AAMAS-2021 Formatting Instructions]{The Tight Bound for Pure Price of Anarchy in an Extended Miner's Dilemma Game}

\author{Qian Wang}
\affiliation{
  \institution{Peking University}
  \city{Beijing}}
\email{charlie@pku.edu.cn}

\author{Yurong Chen}
\affiliation{
  \institution{Peking University}
  \city{Beijing}}
\email{chenyur911@pku.edu.cn}

\begin{abstract}

Pool block withholding attack is performed among mining pools in digital cryptocurrencies, such as Bitcoin. Instead of mining honestly, pools can be incentivized to infiltrate their own miners into other pools. These infiltrators report partial solutions but withhold full solutions, share block rewards but make no contribution to block mining. The block withholding attack among mining pools can be modeled as a non-cooperative game called ``the miner's dilemma'', which reduces effective mining power in the system and leads to potential systemic instability in the blockchain. However, existing literature on the game-theoretic properties of this attack only gives a preliminary analysis, e.g., an upper bound of $3$ for the pure price of anarchy ($\mathrm{PPoA}$) in this game, with two pools involved and no miner betraying. Pure price of anarchy is a measurement of how much mining power is wasted in the miner's dilemma game. Further tightening its upper bound will bring us more insight into the structure of this game, so as to design mechanisms to reduce the systemic loss caused by mutual attacks. In this paper, we give a tight bound of $(1,2]$ for the pure price of anarchy. Moreover, we show the tight bound holds in a more general setting, in which infiltrators may betray. We also prove the existence and uniqueness of pure Nash equilibrium in this setting. Inspired by experiments on the game among three mining pools, we conjecture that similar results hold in the N-player miner's dilemma game ($N\geq 2$). 

\end{abstract}


\newcommand{\BibTeX}{\rm B\kern-.05em{\sc i\kern-.025em b}\kern-.08em\TeX}

\usepackage{subfigure}

\begin{document}


\pagestyle{fancy}
\fancyhead{}


\maketitle 


\section{Introduction}









Bitcoin, assumed to be one of the most successful applications of blockchain, has gained considerable attention since its inception in 2008 \cite{nakamoto2008peer}. One research direction of interest is to study its security and potential attacks. Bitcoin's security mainly comes from two sources: the blockchain data structure and proof-of-work (PoW) consensus protocol. Blockchain is an open, transparent, decentralized digital ledger that can validate transactions between two parties without third-party authentication. Transaction records are stored as blocks linked by hash pointers. PoW is used to decide who gets the power of authorizing the next valid block. In this protocol, miners have to solve a computationally difficult puzzle, and the first miner working out the solution can announce his block and get the block reward. This difficult puzzle is specially designed so that miners spending more mining power have larger probabilities to solve it. Due to fierce competition in Bitcoin system, it may take months, even years, for a single miner to find the solution, thus miners tend to form mining pools to reduce the high variance of mining rewards.


In a mining pool, miners will work on the same puzzle together. Once a miner successfully finds a full solution and submits it, the pool manager is responsible for reporting the block and splitting the reward following certain fair allocation rule. Each miner will get the reward proportional to the mining power he spends. To evaluate how much power miners spend, a pool manager will accept blocks with lower difficulty from miners, called ''share'' or partial solution, as their proof of work, and treat the number of shares handed in as their estimated mining power. 

Such open pools are susceptible to the pool block withholding attack \cite{eyal2015miner}. \citet{eyal2015miner} demonstrated that mining pools have the incentive to infiltrate their own miners into other opponent pools. These infiltrators only submit partial solutions while throw away valid blocks, thus they can share block rewards but make no contribution to block mining. The attack among mining pools can be modeled as a non-cooperative game. \citet{eyal2015miner} referred this game with two players as ``the 2-player miner's dilemma''. In a 2-player miner's dilemma, two open mining pools choose the number of mining power to attack each other. While mining pools performing this attack lose their own mining power, resulting in the decrease of their expected reward from mining, the overall reward can exceed that from honest mining. \citet{alkalay2019pure} gave a detailed analysis based on this model. It calculated the Nash equilibrium in special cases and proved that the pure price of anarchy ($\mathrm{PPoA}$), which measures how much mining power is wasted due to the attack, is at most 3 in general case. It also conjectured the pure Nash equilibrium of the game is unique and the upper bound of $\mathrm{PPoA}$ is 2.

We advance this game analysis further in this paper by proving the conjecture proposed in \citet{alkalay2019pure}. That is, we prove the existence and uniqueness of the Nash equilibrium, and show the tight bound of $\mathrm{PPoA}$, defined as in \citet{alkalay2019pure}, is $(1,2]$. We provide conditions to determine when this game admits an extreme equilibrium or a non-extreme equilibrium and give an algorithm to calculate it.

Moreover, we actually prove the results in an extended model, which includes previous model from \citet{eyal2015miner,alkalay2019pure} as a special case. Instead of assuming the loyalty of miners as in \citet{eyal2015miner} and \citet{alkalay2019pure}, we allow infiltrators to betray for their own interests, since they can gain more rewards by reporting full solutions without letting the original pool know. The rationality of this betrayal assumption is discussed in detail in Section \ref{betrayal_section}. We will also introduce a betrayal parameter to indicate the percentage of betraying mining power. The game reduces to the original miner's dilemma when no infiltrator betrays.

Overall, our main contributions are as follows:
\begin{enumerate}
 \item We adopt the betrayal assumption and extend the previous 2-player miner's dilemma game to a  more general model. (Section \ref{model_section})
 \item We prove the existence and uniqueness of pure Nash equilibrium of this game. (Theorem \ref{existence} and Theorem \ref{uniqueness})
 \item We show the tight bound of pure price of anarchy ($\mathrm{PPoA}$) is $(1, 2]$. (Theorem \ref{main})
 \item We conduct experiments on the game among three pools, which provide convincing evidence for our conjecture that $N$-player ($N\geq 2$) miner's dilemma game admits a unique pure Nash equilibrium and $\mathrm{PPoA}$ is within $(1,2]$. (Section \ref{experiment_section})
\end{enumerate}

Our paper is organized as follows. In Section \ref{model_section}, we introduce the extended model of 2-player miner's dilemma with betrayal assumption. Then in Section \ref{existence_section}, we prove the existence of pure Nash equilibrium in this game. Considering that the player's strategy space is bounded, we call an equilibrium point taken at the boundary point an extreme equilibrium, and in the middle a non-extreme equilibrium. We discuss these two cases in Section \ref{extreme_section} and Section \ref{non-extreme_section} separately. With the help of the characterization of Nash equilibrium in these two sections, we can prove the uniqueness of Nash equilibrium in the following Section \ref{uniqueness_section}. Finally in Section \ref{ppoa_section}, we use the results in previous sections to induce the bound of $\mathrm{PPoA}$. Our experimental results and conjecture are illustrated in Section \ref{experiment_section}.
 
 
 
 

\section{Related Work}





Block withholding attack was first proposed in \citet{rosenfeld2011analysis} that a miner in a mining pool may have the incentive to withhold valid blocks or delay block submissions. \citet{luu2015power} showed a single attacker can get higher rewards than honestly mining by performing this attack. The attacker could be not only an individual who controls a certain amount of mining power, but also another mining pool. \citet{ke2019ibwh} took the difficulty adjustment per 14 days into consideration and proposed a novel adversarial strategy for reward rate named the intermittent block withholding attack (IBWH). In the model proposed by \citet{qin2020optimal}, the probability of generating a solution per unit mining power of each pool to be a complete solution varies among pools.

\citet{eyal2015miner} further considered the case when two pools perform attacks on each other and discovered the miner’s dilemma. Each mining pool is trying to optimize its own average reward, but ending with the loss of mining power in the Nash equilibrium. \citet{eyal2015miner} briefly claimed that to calculate a Nash equilibrium in a 2-player miner's dilemma, we only need to solve the vanishing partial derivatives using symbolic computation tools. However, \citet{alkalay2019pure} showed the claim is incorrect and provided a detailed analysis. They also proved the existence of pure Nash equilibrium in 2-player miner's dilemma game. For symmetric case and no-other-miner case, they characterized the Nash Equilibrium, and for general case, they gave an upper bound $3$ of the game's pure price of anarchy, specially defined in its paper. 

\citet{karpova2019game}, based on the result by \citet{alkalay2019pure}, used experimental simulations to see how the number of mining pools affects the convergence. They also extended the simulation by allowing miners to switch between pools or to choose solo mining. \citet{li2020mining} established a model with system rewards and punishments and analyzed the Nash equilibrium. They also considered the betrayal rate, but the definition of betrayal is different from ours. In \citet{li2020mining}, a betrayer will pledge their loyalty to the opponent pool directly. \citet{haghighat2019block} modeled pool block withholding attack as a stochastic game and used reinforcement learning to see its evolution.

Our model mainly follows \citet{eyal2015miner} and \citet{alkalay2019pure}. The mining efficiencies are the same between pools and the optimization goal is the average reward rather than the reward rate. The only difference is that we introduce a betrayal parameter to model infiltrators' betrayal. 
\section{Model}
\label{model_section}


In this section, we give a formal description of our model. We consider the game between two mining pools as two players. Each mining pool can choose how much mining power will be sent to the other pool as its strategy. After specifying notations and basic assumptions of our model, we present our betrayal assumptions in detail and express the reward functions for the two pools to be optimized.

We assume the total mining power of the system and the mining power of each pool are fixed. Let $m$ denote the total mining power of the system and let $m_i$ denote the mining power of pool $i$, $i\in \{1, 2\}$. The values of $m, m_1, m_2$ should all be positive. There may be other mining power outside the game, solo miners or other mining pools, but we assume they have no interaction with these two pools. We denote this left part of mining power as $t$ and $t = m - m_1 - m_2 \ge 0$. At some steps of analysis, we shall replace $m$ with $m_1+m_2+t$ without notice. Let $x_i$ denote the amount of the mining power used by pool $i$ to attack the other, $x_i\in [0, m_i]$. Thus, a pure strategy profile is $(x_1, x_2)$.

\subsection{Betrayal Assumption}
\label{betrayal_section}


Although the mining power infiltrating the opposite is essentially a number of machines, they are still controlled by human miners. Pool managers need to consider miner's personal interests at the micro level when making macro decisions. Specifically, infiltrators may betray the original pool for their own interests, which means they may report full solutions without telling the original pool. 

This report process can be secret, so that the original mining pool cannot distinguish between betraying miners and loyal miners. More importantly, since full solutions are also counted as shares, betraying miners can get more reward from the opponent mining pool and hide the extra reward for himself. Although the extra reward could be negligible, even $0$, it is always non-negative, and positive in expectation, so infiltrators do have motives to betray.

What needs to be emphasized is that betraying, in our setting, is different from joining the other pool directly. Joining the other pool directly means the miner will only participate in the mining reward distribution of the opponent pool. While Betraying means that the miner still participates in the reward distribution twice like other non-betraying infiltrators, but after receiving the reward from the opponent pool, he can secretly hide the unclaimed part and then participate in the reward distribution of the original pool. Actually, if infiltrators join the other pool directly, it is equivalent to the game assuming loyalty with different $m_i$ initially.

Considering that different miners own different moral thresholds and that some miners are not even aware he is an infiltrator \cite{courtois2014longest}, not all infiltrators will betray. Here we introduce a betrayal parameter, $p\in [0, 1]$, to represent the percent of betrayal. The effective mining power of pool $i$ attacking the other pool is actually $(1-p)x_i$. For example, if there are $100$ miners controlling identical mining machines and $p = 0.4$, $40$ miners will choose to betray. Notice the model employed in previous work \cite{eyal2015miner,alkalay2019pure} is a special case of ours when $p=0$.

\subsection{Reward Functions}

Given the above notations, we are ready to express the optimization goals for both sides. Since the mining power of each pool is fixed, maximizing the total reward and maximizing the average reward (reward per unit mining power) are equivalent. 

The total effective mining power of the system is $m-(1-p)(x_1+x_2)$. The effective mining power of mining pool $1$ is $m_1-x_1 + px_2$, and that of mining pool $2$ is $m_2 - x_2 + px_1$. The direct reward of pool $i$ from the Bitcoin system, denoted as $R_i(x_1, x_2)$, is proportional to the fraction of the effective mining power contributed to the system by the pool.

$$R_{1}\left(x_{1}, x_{2}\right)=\frac{m_{1}-x_{1} + px_{2}}{m-(1-p)(x_{1}+x_{2})},$$
$$R_{2}\left(x_{1}, x_{2}\right)=\frac{m_{2}-x_{2}+px_{1}}{m-(1-p)(x_{1}+x_{2})}.$$

If $m-(1-p)(x_1+x_2) = 0$, it means $x_1 = m_1$, $x_2 = m_2$ and $t = p = 0$. Nobody can get reward because nobody is mining. Obviously, either pool can retreat a positive amount of mining power to get a positive reward. $(m_1, m_2)$ cannot be a Nash equilibrium, so that we can always assume $m-(1-p)(x_1+x_2)>0$ and the above functions are well defined.

In addition to direct rewards, each pool will get reward from infiltrating the other pool, which should be the product of the average reward of the other pool and the amount of infiltrating mining power\footnote{Note this is only an approximation. The real value will be less than it because infiltrators do not provide full solutions, so less shares lead to less rewards. But the difference is negligible while calculating, especially for the simplicity of the model. No previous work has considered to be more precise. The same goes for the reward hidden by the betrayers, which can be ignored from a macro perspective.}. Let $r_i(x_1, x_2)$ denote the average reward 
of pool $i$, $i \in \{1,2\}$, we have 

$$r_{1}\left(x_{1}, x_{2}\right) =\frac{R_{1}\left(x_{1}, x_{2}\right)+x_{1} r_{2}\left(x_{1}, x_{2}\right)}{m_{1}+x_{2}},$$
$$r_{2}\left(x_{1}, x_{2}\right) =\frac{R_{2}\left(x_{1}, x_{2}\right)+x_{2} r_{1}\left(x_{1}, x_{2}\right)}{m_{2}+x_{1}}.$$

After substituting $R_{1}\left(x_{1}, x_{2}\right), R_{2}\left(x_{1}, x_{2}\right)$ with their expressions and solving the above system of equations, we get

$$r_{1}\left(x_{1}, x_{2}\right)=\frac{m_{1} m_{2}+m_{1} x_{1}+ pm_{2}x_{2}-(1-p)x_{1}^{2}-(1-p)x_{1} x_{2} }{\left(m-(1-p)(x_1+x_2)\right)\left(m_{1} m_{2}+m_{1} x_{1}+m_{2} x_{2}\right)},$$
$$r_{2}\left(x_{1}, x_{2}\right)=\frac{m_{1} m_{2}+m_{2} x_{2}+ pm_{1}x_{1} -(1-p)x_{2}^{2}-(1-p)x_{1} x_{2} }{\left(m-(1-p)(x_1+x_2)\right)\left(m_{1} m_{2}+m_{1} x_{1}+m_{2} x_{2}\right)}.$$

Again, the total reward to pool i is $m_ir_i(x_1, x_2)$, but it is enough to only consider the average reward as $m_i$ is fixed. 

Note $r_{1}\left(x_{1}, x_{2}\right)=r_{2}\left(x_{1}, x_{2}\right)\equiv\frac{1}{m}$ when $p=1$, which means however much mining power is used to attack, the average reward will remain the same if every infiltrator betrays. In this case, every strategy profile $(x_1,x_2)$ is a pure Nash Equilibrium. Thus we only consider $p\in[0,1)$.


\section{Existence of Nash Equilibrium}
\label{existence_section}


We only focus on the pure Nash Equilibrium in this work. A strategy profile is a pure Nash equilibrium if each player cannot improve his reward by changing his strategy. In this game, $\left(x_{1}^{*}, x_{2}^{*}\right)$ is a pure Nash Equilibrium if the average reward of each mining pool is maximized given the other's strategy, i.e., if

$$x_{1}^{*}=\underset{x_{1} \in\left[0, m_{1}\right]}{\arg \max } r_{1}\left(x_{1}, x_{2}^{*}\right), \quad x_{2}^{*}=\underset{x_{2} \in\left[0, m_{2}\right]}{\arg \max } r_{2}\left(x_{1}^{*}, x_{2}\right).$$

We provide the first order partial derivatives of two reward functions.
\begin{small}
\begin{equation*}
\begin{aligned}
\partial_{x_{1}} r_{1}\left(x_{1}, x_{2}\right)&=
\frac{
\begin{array}{cc}

(1-p)\left\{(1-p)\left[m_{2} x_{2}\left(x_{1}+x_{2}\right)^{2} +  m_{1} m_{2} x_{1}^{2}\right]\right. \\
+ m_{1}^{2}\left(m_{2}+x_{1}\right)^{2} + pm_2^2x_2^2 - mm_{2} x_{2}\left(2 x_{1}+x_{2}\right)\\
+ \left.m_{1}\left[m_{2}\left((1+p)m_{2} x_{2}+2 x_{1} x_{2}\right)-m x_{1}\left(2 m_{2}+x_{1}\right)\right]\right\}
\end{array}
}{
\begin{array}{cc}
\left(m-(1-p)(x_{1}+x_{2}) \right)^{2}\left(m_{1} m_{2}+m_{1} x_{1}+m_{2} x_{2}\right)^{2}
\end{array}
}\\
\partial_{x_{2}} r_{2}\left(x_{1}, x_{2}\right)&=
\frac{
\begin{array}{cc}

(1-p)\left\{(1-p)\left[m_{1} x_{1}\left(x_{1}+x_{2}\right)^{2} +  m_{1} m_{2} x_{2}^{2}\right]\right. \\
+ m_{2}^{2}\left(m_{1}+x_{2}\right)^{2} + pm_1^2x_1^2 - mm_{1} x_{1}\left(2 x_{2}+x_{1}\right)\\
+ \left.m_{2}\left[m_{1}\left((1+p)m_{1} x_{1}+2 x_{1} x_{2}\right)-m x_{2}\left(2 m_{1}+x_{2}\right)\right]\right\}
\end{array}
}{
\begin{array}{cc}
\left(m-(1-p)(x_{1}+x_{2})\right)^{2}\left(m_{1} m_{2}+m_{1} x_{1}+m_{2} x_{2}\right)^{2}
\end{array}
}
\end{aligned}
\end{equation*}
\end{small}

\begin{lemma}
\label{concavity_lemma}
$r_i(x_1,x_2)$ is concave with respect to $x_i$, $i \in \{1,2\}$. 
\end{lemma}

\begin{proof}[proof of lemma \ref{concavity_lemma}]
We only prove the concavity of $r_1(x_1,x_2)$ as the proof of $r_2(x_1,x_2)$ is symmetric. We organize the numerator of $\partial^2_{x_1}r_1(x_1,x_2)$ as a quadratic polynomial in $t$, $$\sum_{i=0}^2 c_i(m_1, m_2, p, x_1, x_2)t^i,$$ where $c_i(\cdot)$ is the coefficient of the $i$-order term, $i \in \{0, 1, 2\}$.


It is easy to check $c_1(\cdot)$ and $c_2(\cdot)$ are negative. To prove $r_i(x_1,x_2)$ has a non-positive second derivative, it is sufficient to show $c_0(\cdot) \leq 0$ since $t\geq 0$. Proving $c_0(\cdot) \leq 0$ may require some efforts. Denote $\frac{c_0(\cdot)}{2m_2(1-p)}$ as $h(m_1, m_2)$, $m_1\in[x_1,+\infty]$, $m_2\in[x_2,+\infty]$. By computing the higher derivatives of $h(\cdot)$ with respect to $m_2$ then $m_1$, we find $h(m_1, m_2)$ achieves its maximum when $m_1 = x_1$ and $m_2 = x_2$.

$$h(m_1, m_2) \leq h(x_1, x_2) \leq 0 $$

Therefore we have $\partial^2_{x_1}r_1(x_1,x_2) \leq 0$ and the proof is complete.

\end{proof}

The following Theorem \ref{existence} proves the existence of the Nash equilibrium using the concavity, and gives the sufficient and necessary condition for a strategy profile $(x^*_1,x^*_2)$ to be a Nash equilibrium.

\begin{theorem}
\label{existence}
Every two-player miner’s dilemma game with betrayal assumption admits at least one pure Nash equilibrium. For each Nash equilibrium $(x_1^{*}, x_2^{*})$, $x_1^{*}$ meets one and only one of the following conditions; a symmetric statement also holds for $x_{2}^{*}$.

(a) $\partial_{x_1} r_{1}\left(x_{1}, x_{2}^{*}\right) \leq \partial_{x_1} r_{1}\left(0, x_{2}^{*}\right) \leq 0\ \land\ x_{1}^{*} = 0$;

(b) $\partial_{x_1} r_{1}\left(x_{1}, x_{2}^{*}\right) \geq \partial_{x_1} r_{1}\left(m_1, x_{2}^{*}\right) \geq 0\ \land\ x_{1}^{*} = m_1$;

(c) $\partial_{x_1} r_{1}\left(x_{1}^{*}, x_{2}^{*}\right) = 0\ \land\ x_{1}^{*}\in (0, m_1)$.

\end{theorem}

\begin{proof}[proof of theorem \ref{existence}]

As \citet{alkalay2019pure} has shown, When $t=0$ and $p=0$, the theorem is true and the unique Nash equilibrium is given by

\[
\left(x_{1}^{*}, x_{2}^{*}\right)=\left\{\begin{array}{ll}
\left(0, \frac{m_{2}}{2}\right) & \text { if } m_{1} \leq \frac{m_{2}}{4}, \\
\left(\frac{m_{1}}{2}, 0\right) & \text { if } m_{2} \leq \frac{m_{1}}{4}, \\
\left(\frac{\sqrt{m_{1} m_{2}}(2 \sqrt{m_{1}}-\sqrt{m_{2}})}{\sqrt{m_{1}}+\sqrt{m_{2}}}, \right.\\
\left.\frac{\sqrt{m_{1} m_{2}}(2 \sqrt{m_{2}}-\sqrt{m_{1}})}{\sqrt{m_{1}}+\sqrt{m_{2}}}\right) & \text { otherwise. }
\end{array}\right.
\]

When $t>0$ or $p>0$, $m-(1-p)(x_1+x_2)>0$ holds, so reward functions are always well defined in $[0, m_1]\times[0, m_2]$. Then we can leverage the result by \citet{glicksberg1952further}. In our model, the strategy sets $[0, m_1]$ and $[0, m_2]$ are nonempty, compact and convex. For $i \in \{1, 2\}$, $r_i(x_1, x_2)$ is continuous in $[0, m_1]\times[0, m_2]$ and concave in $x_i$ by Lemma \ref{concavity_lemma}. Three conditions in Glicksberg's Existence Theorem are all satisfied, so the existence of pure Nash equilibrium is guaranteed. 

The concavity of reward functions also indicates the left part of theorem is true.

\end{proof}

In fact, the pure Nash equilibrium of this game is unique, but we will elaborate it later in Section \ref{uniqueness_section}. Before that, we need to characterize the Nash equilibrium as a preparation for the proof of uniqueness.

\section{Extreme Equilibrium}
\label{extreme_section}

We start from condition (a) and (b) in Theorem \ref{existence}, to see when $x_i=0$ or $x_i=m_i$ becomes a Nash equilibrium and how the Nash equilibrium will look like. We show $x_i=0$ is feasible under some constraints and show the corresponding best strategy of the other, while using all mining power to infiltrate will never be the equilibrium strategy. 

\begin{lemma}
\label{extreme0}
The unique pure Nash equilibrium is given by

\[
\begin{small}
\left(x_{1}^{*}, x_{2}^{*}\right)=\left\{\begin{array}{ll}
\left(0, \frac{m_{2}}{2}\right) & \text { if } t = 0, p = 0, \\
\left(0, \begin{small}\frac{m_{1}\left(m_{2}-m+\sqrt{m^{2}-mm_{2}-(1-p)m_{1} m_{2}}\right)}{m-(1-p)m_{1}-m_{2}}\end{small}\right) & \text { otherwise, }
\end{array}\right.
\end{small}
\]
if and only if $g(t, p, m_1, m_2) \leq 0$, where g is a function defined by
\begin{equation*}
\begin{small}
\begin{array}{rl}
g(t, p, m_1, m_2) =& 4t^3 \\
&+ 4(3m_1+m_2)t^2 \\
&+ \left[12m_1^2 + (2+8p-2p^2)m_1m_2 - (1-p)^2m_2^2\right]t \\
&+ m_1(m_1 + pm_2)(4m_1 - (1-p)^2m_2).
\end{array}
\end{small}
\end{equation*}
Moreover, we have $x^*_1+x^*_2<\frac{m_1+m_2}{2}$.

A symmetric statement holds if $g(t, p, m_2, m_1) \leq 0$. In that case, $x_2^{*}=0$ and notations $m_1$ and $m_2$ swap in the corresponding $x_1^{*}$.
\end{lemma}

To deliver a better understanding of the function $g$ in Lemma \ref{extreme0}, we give the following equivalent condition.
\begin{equation*}
\begin{small}
g(t, p, m_1, m_2) \leq 0 \Leftrightarrow m_1 \leq \frac{(1-p)^2}{4}m_2\ \land\ t \in [0, t^{*}(p, m_1, m_2)],
\end{small}
\end{equation*}
where $t^{*}(p, m_1, m_2)$ is the unique positive solution to 

$$g(t, p, m_1, m_2) = 0, \  m_1 \leq \frac{(1-p)^2}{4}m_2.$$

Thus intuitively, this game admits $(0, x_2^{*})$ or $(x_1^{*}, 0)$ as a Nash equilibrium if and only if one pool owns much higher mining power than the other and the mining power outside the two mining pools is small.

\begin{proof}[proof of lemma \ref{extreme0}]

We only need to prove one side when $x_1^{*}=0$. Let $(0,x_2^*)$ be a pure Nash equilibrium, then

\begin{equation*}
\begin{small}
   \partial_{x_2}r_2(0,x_2)=\frac{m^2_2(m_1+x_2)^2+m_2\left[(1-p)m_1x^2_2-mx_2(2m_1+x_2)\right]}{(m-x_2)^2(m_1m_2+m_2x_2)^2/(1-p)}
\end{small}
\end{equation*}

The denominator is always positive and the numerator is a quadratic polynomial in $x_2$, the quadratic term of which is $m_2(m_2+(1-p)m_1-m)$. If the quadratic term is 0, i.e., $t=0$ and $p=0$, we can yield $x^*_2=\frac{m_2}{2}$ by solving $ \partial_{x_2}r_2(0,x_2)=0$. Otherwise, solving the quadratic equation, we get $x_2=\frac{m_{1}\left(m_{2}-m\pm\sqrt{m^{2}-mm_{2}-(1-p)m_{1} m_{2}}\right)}{m-(1-p)m_{1}-m_{2}}$. Only the positive root is within $[0,m_2]$. Notice $ \partial_{x_2}r_2(0,x_2)=0$ already yields a solution, so neither $x_2^*=0$ nor $x_2^*=m_2$ can be feasible.

When $t>0$ or $p>0$, solving $$ \partial_{x_1}r_1(0,\frac{m_{1}\left(m_{2}-m+\sqrt{m^{2}-mm_{2}-(1-p)m_{1} m_{2}}\right)}{m-(1-p)m_{1}-m_{2}}) \leq 0, $$ we get $g(t, p, m_1, m_2) \leq 0$ by symbolic calculation. When $t=0$ and $p=0$, using $x^*_2=\frac{m_2}{2}$ we get $m_1\leq\frac{m_2}{4}$, which also leads to $g(0, 0, m_1, m_2) \leq 0$.

Now we prove $x^*_1+x^*_2<\frac{m_1+m_2}{2}$. When $t=0$ and $p=0$, it is trivial to show $x^*_1+x^*_2 = \frac{m_2}{2} <\frac{m_1+m_2}{2}$. When $t>0$ or $p>0$, it is equivalent to prove
\begin{equation*}
    \frac{2m_1\left(-t-m_1+\sqrt{(m_1+t)^2+m_2(t+pm_1)} \right)}{t+pm_1}< m_1+m_2
\end{equation*}
\begin{equation*}
    \Longleftrightarrow\sqrt{1+\frac{m_2(t+pm_1)}{(m_1+t)^2}}<\frac{(m_1+m_2)(t+pm_1)}{2m_1(m_1+t)}+1,
\end{equation*}

and this inequality holds for
\begin{equation*}
    \begin{aligned}
        \left(1+\frac{(m_1+m_2)(t+pm_1)}{2m_1(m_1+t)}\right)^2 &> 1+\frac{(m_1+m_2)(t+pm_1)}{m_1(m_1+t)}\\
        & >1+\frac{m_2(t+pm_1)}{(m_1+t)^2}.
    \end{aligned}
\end{equation*}
\end{proof}

\begin{lemma}
\label{extremeM}
$\left(x_{1}^{*}, x_{2}^{*}\right)$ cannot be a pure Nash equilibrium if $x_{1}^{*}=m_{1}$ or $x_{2}^{*}=m_{2}$.
\end{lemma}

We have mentioned the strategy profile $(m_1, m_2)$ cannot be a Nash equilibrium when introducing the reward functions. Lemma \ref{extremeM} further shows neither pool will infiltrate all the mining power into the opponent pool. Actually, this is so unreasonable a choice that we can easily find a better strategy. In the proof of Lemma \ref{extremeM}, we compare the rewards between infiltrating no mining power and infiltrating it all.

\begin{proof}[proof of lemma \ref{extremeM}]

We only need to prove one side when $x_1^{*}=m_1$. We first prove that for any $x_2 \in (0, m_2)$, $r_1(0,x_2)>r_1(m_1,x_2)$.

\begin{equation*}
    \begin{aligned}
    r_1(0,x_2)&=\frac{m_1m_2+pm_2x_2}{(m-(1-p)x_2)(m_1m_2+m_2x_2)},\\
    r_1(m_1,x_2)&=\frac{m_1m_2+m_1^2+pm_2x_2-(1-p)m_1^2-(1-p)m_1x_2}{(m-(1-p)(m_1+x_2))(m_1m_2+m_1^2+m_2x_2)},
    \end{aligned}
\end{equation*}
and $r_1(0,x_2)>r_1(m_1,x_2)$ is equivalent to 
\begin{equation*}
    -(1-p)x^3_2+(t+pm_1+(1-p)m_2)x^2_2+m_1(t-pm_2)x_2+m^2_1t > 0.
\end{equation*}

Denote the left side of the above inequality as function $\phi(x_2):[0,m_2] \rightarrow \mathbb{R}$, then $\phi(0)=m^2_1t\geq 0$ and $\phi(m_2)=t( m^2_1+m_1m_2+m^2_2)\geq 0$. By derivative analysis, the minimum of $\phi$ in interval $[0,m_2]$ can only be achieved at $0$ or $m_2$. So $$\phi(x_2) > \max(\phi(0), \phi(m_2)) \geq 0.$$

In addition, we have already removed $(m_1, m_2)$ from the very beginning and $(m_1, 0)$ cannot be a Nash equilibrium by Lemma \ref{extreme0}. Hence, $\forall x_2 \in [0, m_2]$, $(m_1, x_2)$ is not a Nash equilibrium and the proof is completed.

\end{proof}

By Theorem \ref{existence}, Lemma \ref{extreme0} and Lemma \ref{extremeM}, we immediately have the following corollary.

\begin{corollary}
\label{add_a_corollary}
Every two-player miner’s dilemma game with betrayal assumption admits at least one pure Nash equilibrium. For each Nash equilibrium $(x_1^{*}, x_2^{*})$, one and only one of the following conditions is met: (a) $x_1^{*} = 0$ ; (b) $x_2^{*} = 0$ ; (c) $(x_1^{*}, x_2^{*})$ is a solution of vanishing partial derivatives within $(0, m_1)\times(0, m_2)$, i.e.,
\begin{equation*}
\begin{array}{rl}
    \partial_{x_{1}} r_{1}\left(x_{1}^{*}, x_{2}^{*}\right) = 0,&\ \partial_{x_{2}} r_{2}\left(x_{1}^{*}, x_{2}^{*}\right) = 0,\\
\textit{s.t.}\quad x_{1}^{*}\in (0, m_1),&\ x_{2}^{*}\in (0, m_2).
\end{array}
\end{equation*}
\end{corollary}

\section{Non-extreme Equilibrium}
\label{non-extreme_section}

In this section, we restrict our attention to condition (c) in Corollary \ref{add_a_corollary}. We will prove the vanishing partial derivatives yield at most one feasible solution $(x_1^{*}, x_2^{*})$ within $(0, m_1)\times(0, m_2)$, and $x_1^{*}+x_2^{*} \leq \frac{m_1+m_2}{2}$. In special cases, we also express the pure Nash equilibrium and show when $x_1^{*}+x_2^{*} = \frac{m_1+m_2}{2}$ holds.

Simple calculations show $\partial_{x_{i}} r_{i}\left(x_{1}^{*}, x_{2}^{*}\right) = 0$ for $i\in \{1, 2\}$ is equivalent to

\begin{equation*}
    \begin{array}{ll}
        \partial_{x_{1}} r_{1}\left(x_{1}^{*}, x_{2}^{*}\right) + \partial_{x_{2}} r_{2}\left(x_{1}^{*}, x_{2}^{*}\right) & = 0,  \\
        m_1(m_2 + 2x_1^{*})\partial_{x_{1}} r_{1}\left(x_{1}^{*}, x_{2}^{*}\right) - m_2(m_1 + 2x_2^{*})\partial_{x_{2}} r_{2}\left(x_{1}^{*}, x_{2}^{*}\right) & = 0.
    \end{array}
\end{equation*}

The second equation above seems unnatural, but it will lower the degree to make the final expansion simpler. Substituting the partial derivatives, the above equation system is equivalent to

\begin{equation}
\label{trick}
\begin{array}{ll}
(1-p)\left(x_{1}^{*}+x_{2}^{*}\right)^2 -2t\left(x_{1}^{*}+x_{2}^{*}\right) & \\
+ 2(m_{1}m_{2}-m_{1}x_{2}^{*} - m_{2}x_{1}^{*}) - (1-p)(m_{1}x_{1}^{*} + m_{2}x_{2}^{*}) & =0,\end{array}
\end{equation}

\begin{equation}
\label{trick2}
\begin{array}{rc}
m_1m_2[(1 + p + \frac{2t}{m_1})x_2^{*2} + (2m_1 - (1-p)m_2 + 2t)x_2^* & \\- (1 + p + \frac{2t}{m_2})x_1^{*2} - (2m_2 -(1-p)m_1 + 2t)x_1^{*}] & =0.\end{array}
\end{equation}

Let $y^{*} = x_1^{*} + x_2^{*}$ and replace $x_2^{*}$ with $y^{*} - x_1^{*}$ in Equation (\ref{trick}). We get

\begin{equation}
\label{x1,y}
x_1^{*} = \frac{(1-p)y^{*2} - (2m_1+(1-p)m_2+2t)y^{*} + 2m_1m_2}{(m_2 - m_1)(1+p)} \end{equation}

Substituting Equation (\ref{x1,y}) and $x_2^{*} = y^{*} - x_1^{*}$ into Equation (\ref{trick2}), we can get an equation with only one variable $y^*$.

\begin{equation}
\label{f=0}
\begin{array}{ll}
0 =& (1-p)^2t(y^{*})^4 \\
&- (1-p)\left[4t^2 + 4(m_1+m_2)t + (1+p)^2m_1m_2\right](y^{*})^3  \\
&+ \left[4t^3 + 8(m_1+m_2)t^2\right. \\
&\quad+ (4m_1^2 + 4m_2^2 + 11m_1m_2+ 3p^2m_1m_2-2pm_1m_2)t \\
&\quad \left. + (1+p)^2m_1m_2(m_1+m_2)\right](y^{*})^2 \\
&+(1-p)m_1m_2\left[(1+p)^2m_1m_2 - 4(m_1+m_2)t - 4t^2\right]y^{*} \\
&- m_1^2m_2^2\left[(1+p)^2(m1 + m2) +4pt\right]
\end{array}
\end{equation}

Regard the right side of Equation (\ref{f=0}) as a function $f(y):(0, m_1+m_2)\rightarrow \mathbb{R}$. If we can solve $f(y^*) = 0$, then we can calculate $(x_1^*, x_2^*)$ and check whether it is a legal solution.

\subsection{Symmetric Case}

Before we study function $f$, note that in Equation (3), the denominator may be $0$ in the right side of Equation (\ref{x1,y}). So we have to consider the symmetric case separately, i.e., $m_1=m_2$. Lemma \ref{symmetric} shows if two mining pools have the same mining power, they will also choose the same strategy, which is consistent with our intuition.

\begin{lemma}
\label{symmetric}
If $m_{1}=m_{2}=\frac{m}{k}$ where $k \geq 2$, the unique pure Nash equilibrium $\left(x_{1}^{*}, x_{2}^{*}\right)$ is given by
\[
x_{1}^{*}=x_{2}^{*}=\frac{m}{4 (1-p) k}(2 k-(1+p)-\sqrt{(2k-(p+1))^2 - 8(1-p)}).
\]
$x_1^{*}+x_2^{*} \leq \frac{m_1+m_2}{2}$, and equal to $\frac{m_1+m_2}{2}$ if and only if $m_{1}=m_{2}=\frac{m}{2}$.
\end{lemma}

\begin{proof}[proof of lemma \ref{symmetric}]

When $m_1=m_2=\frac{m}{k}$, $\partial_{t} g(t, p, \frac{m}{k}, \frac{m}{k}) = 12t^2 + \frac{32m}{k}t + (13 + 10p - 3p^2)\left(\frac{m}{k}\right)^2 > 0$. So $g(t, p, \frac{m}{k}, \frac{m}{k}) > g(0, p, \frac{m}{k}, \frac{m}{k}) = \left(\frac{m}{k}\right)^3(1+p)[4-(1-p)^2] > 0$. By Lemma \ref{extreme0}, there is no extreme Nash equilibrium.

Substituting $m_1=m_2=\frac{m}{k}$ into Equation (\ref{trick2}), we obtain
\begin{equation*}
    \left(\frac{m}{k}\right)^2\left(1+p+\frac{2kt}{m}\right)\left(x_{1}^{*}+x_{2}^{*} + \frac{m}{k}\right)\left(x_{2}^{*} - x_{1}^{*}\right) = 0.
\end{equation*}

Hence we must have $x_{1}^{*} = x_{2}^{*}$. Substituting this and $m_1=m_2=\frac{m}{k}$ into Equation (\ref{trick}), we get

\begin{equation*}
    2(1-p)(x_{2}^{*})^2 - \left(2t+(3-p)\frac{m}{k}\right)x_{2}^{*} + \left(\frac{m}{k}\right)^2 = 0.
\end{equation*}

Note $t = m - m_1 - m_2 = \left(1-\frac{2}{k}\right)m$. Solving this quadratic equation, we get 

\begin{equation*}
    x_{2}^{*}=\frac{m}{4 (1-p) k}(2 k-(1+p)\pm\sqrt{(2k-(p+1))^2 - 8(1-p)}).
\end{equation*}

The positive root is a increasing function of $k$ and achieves its minimum value at $k=2$.

\begin{equation*}
    \frac{m}{4 (1-p) k}(2 k-(1+p)+\sqrt{(2k-(p+1))^2 - 8(1-p)}) \geq \frac{m}{2(1-p)}.
\end{equation*}

This is a feasible solution only when $p=0$ and $x_{1}^{*} = x_{2}^{*} = \frac{m}{2}$, but we know it is not a Nash equilibrium by Lemma \ref{extremeM}. Hence, the only feasible solution is the negative root and this is a non-extreme Nash equilibrium.

\begin{equation*}
    x_{1}^{*} + x_{2}^{*} = \frac{4m/k}{2 k-(1+p) + \sqrt{(2k-(p+1))^2 - 8(1-p)}} \leq \frac{m}{2}.
\end{equation*}

The equal sign is established if and only if $k=2$.
\end{proof}

\subsection{Binary Case}





Next we consider another special case, binary case. Lemma \ref{binary} provides the Nash equilibrium when there is no mining power outside these two pools. $f(y)$ will degenerates into a solvable cubic function when $t=0$, and we can get a concise explicit expression for $(x_1^*, x_2^*)$.

\begin{lemma}
\label{binary}
Let $t=0$, i.e. $m=m_{1}+m_{2}$. If $m_{1} \leq \frac{(1-p)^2m_{2}}{4}$ or $m_{2} \leq \frac{(1-p)^2m_{1}}{4}$, we will get a unique extreme Nash equilibrium with $x_1^{*}=0$ or $x_2^{*}=0$ respectively, as in Lemma \ref{extreme0}. Otherwise we will get a unique non-extreme equilibrium given by

\[
\begin{aligned}
\left(x_{1}^{*}, x_{2}^{*}\right)=&\left(\frac{\sqrt{m_{1} m_{2}}\left(2 \sqrt{m_{1}}-(1-p) \sqrt{m_{2}}\right)}{(1+p)\left(\sqrt{m_{1}}+\sqrt{m_{2}}\right)}\right., \\
&\left.\frac{\sqrt{m_{1} m_{2}}\left(2 \sqrt{m_{2}}-(1-p) \sqrt{m_{1}}\right)}{(1+p)\left(\sqrt{m_{1}}+\sqrt{m_{2}}\right)}\right).
\end{aligned}
\]
$x_1^{*}+x_2^{*} \leq \frac{m_1+m_2}{2}$, and the equality holds if and only if $m_{1}=m_{2}=\frac{m}{2}$.
\end{lemma}

\begin{proof}[proof of lemma \ref{binary}]

$$g(0, p, m_1, m_2) \le 0 \Longleftrightarrow m_1\leq \frac{(1-p)^2}{4}m_2$$

Thus, $x_1^{*}=0$ when $m_1\leq \frac{(1-p)^2}{4}m_2$ and the value of $x_2^{*}$ will follow Lemma \ref{extreme0}. Symmetrically, we will also get an extreme Nash equilibrium where $x_2^{*}=0$ when $m_2\leq \frac{(1-p)^2}{4}m_1$.

Next, we consider the non-extreme Nash equilibrium. Setting $t=0$ in Equation (\ref{f=0}) , we get

\begin{equation*}
\label{t=0}
(1+p)^2m_1m_2\left((y^{*})^2 - m_1m_2\right)\left(m_1+m_2 - (1-p)y^{*}\right) = 0.
\end{equation*}

This equation has three roots $\{-\sqrt{m_1m_2}, \sqrt{m_1m_2}, \frac{m_1+m_2}{(1-p)}\}$, but only $\sqrt{m_1m_2}$ is a feasible root within $(0, m_1+m_2)$. Substituting $y^{*}=\sqrt{m_1m_2}$ in Equation \ref{x1,y}, we get

\[
x_{1}^{*}=\frac{\sqrt{m_{1} m_{2}})(2 \sqrt{m_{1}}-(1-p)\sqrt{m_{2}})}{(1+p)(\sqrt{m_{1}}+\sqrt{m_{2}})}.
\]
And symmetrically, 

\[
x_{2}^{*}=\frac{\sqrt{m_{1} m_{2}}(2 \sqrt{m_{2}}-(1-p)\sqrt{m_{1}})}{(1+p)(\sqrt{m_{1}}+\sqrt{m_{2}})}.
\]

To be a feasible Nash equilibrium, both $x_1^{*}$ and $x_{2}^{*}$ have be positive. Solving $x_1^{*} > 0$ and $x_{2}^{*}>0$, we get $m_1> \frac{(1-p)^2}{4}m_2$ and $m_2> \frac{(1-p)^2}{4}m_1$. We can verify $x_1^{*} < m_1$ and $x_2^{*} < m_2$ also hold in this condition.

$$x_1^{*}+x_2^{*} = \sqrt{m_1m_2} \leq \frac{m_1+m_2}{2}$$

The equal sign is established if and only if $m_1=m_2=\frac{m}{2}$.

\end{proof}

\subsection{General Case}

Finally, we assume $m_1\neq m_2$, $t>0$ and continue with the general case. $f(y)$ is a quartic function of $y$, which cannot be factorized. Of course we can use root formula to solve it but the result is too complicated to analyze. Instead, we analyze the zero point of $f(y)$ directly. We can determine the existence and range of the zero point by derivative analysis.

\begin{lemma}
\label{f_y}
When $m_1, m_2 > 0$, $m_1 \neq m_2$, $t > 0$ and $p \in [0, 1)$. $f(y)$ has a unique root $y^{*}$ within $(0, m_1+m_2)$ and $y^{*} < \frac{m_1 + m_2}{2}$.
\end{lemma}

\begin{proof}[proof of lemma \ref{f_y}]
We calculate the first and second derivatives of $f(y)$.

\begin{equation*}
\begin{small}
\begin{array}{rl}
f'(y) =& 4(1-p)^2ty^{3} \\
&- 3(1-p)\left[(1+p)^2m_1m_2 + 4(m_1+m_2)t + 4t^2\right]y^{2}  \\
&+ 2\left[4t^3 + 8(m_1+m_2)t^2\right. \\
&\quad + (4m_1^2 + 4m_2^2 + 11m_1m_2+ 3p^2m_1m_2-2pm_1m_2)t \\
&\quad + \left.(1+p)^2m_1m_2(m_1+m_2)\right]y \\
&+(1-p)m_1m_2\left[(1+p)^2m_1m_2 - 4(m_1+m_2)t - 4t^2\right]
\end{array}
\end{small}
\end{equation*}

\begin{equation*}
\label{f}
\begin{small}
\begin{array}{rl}
f''(y) =& 12(1-p)^2ty^{2} \\
&- 6(1-p)\left[(1+p)^2m_1m_2 + 4(m_1+m_2)t + 4t^2\right]y  \\
&+ 2\left[4t^3 + 8(m_1+m_2)t^2\right. \\
&\quad + (4m_1^2 + 4m_2^2 + 11m_1m_2+ 3p^2m_1m_2-2pm_1m_2)t \\
&\quad + \left.(1+p)^2m_1m_2(m_1+m_2)\right]
\end{array}
\end{small}
\end{equation*}

Solving $f''(y) = 0$, we get

\begin{equation*}
\begin{aligned}
y_1 =& \frac{(1+p)^2m_1m_2}{4(1-p)t} + \frac{m_1 + m_2 + t}{1-p} - \frac{\sqrt{\Delta_1}}{12t(1-p)},\\
y_2 =& \frac{(1+p)^2m_1m_2}{4(1-p)t} + \frac{m_1 + m_2 + t}{1-p} + \frac{\sqrt{\Delta_1}}{12t(1-p)},
\end{aligned}
\end{equation*}

where $\Delta_1 = 48t^4 + 96(m_1+m_2)t^3 + 48[(m_1+m_2)^2+4pm_1m_2]t^2 + 48(1+p)^2m_1m_2(m_1+m_2)t + 9(1+p)^4m_1^2m_2^2 > 0$. We can verify that $y_1 \in (0, m_1 + m_2)$ but $ y_2 > m_1 + m_2$. So $f'(y)$ increases on $(0, y_1)$ and decreases on $(y_1, m_1 + m_2)$. Considering the middle point, we have


\begin{equation*}
f'(\frac{m_1+m_2}{2}) > 0.
\end{equation*}

Although we cannnot determine the sign of $f'(0)$ or that of $f'(m_1+m_2)$, we can now claim that there is only one continuous interval containing point $\frac{m_1+m_2}{2}$ where $f'(y)$ is positive, while $f'(y)$ can probably be negative near $0$ and $m_1+m_2$.

Back to the original function $f(y)$, we have


\begin{equation*}
f(\frac{m_1+m_2}{2}) > 0.
\end{equation*}

We know there is only one continuous interval containing point $\frac{m_1+m_2}{2}$ where $f(y)$ increases. Although $f(y)$ may decrease near $0$ and $m_1+m_2$, it does not matter. We can easily verify $f(0) < 0$ and $f(m_1+m_2) > 0$, so there exists a unique root for $f(y)$ within $(0, m_1+m_2)$ and it is stricly less than $\frac{m_1+m_2}{2}$.

\end{proof}

$f(y^*)=0$ is a necessary condition for that the original equation system of vanishing partial derivatives has a solution within $(0, m_1)\times(0, m_2)$. It is insufficient because the corresponding $(x_1^*, x_2^*)$ may be out of range. Lemma \ref{symmetric}, \ref{binary} and \ref{f_y} ensure that given $m_1, m_2, t, p$, there is at most one non-extreme Nash equilibrium. If there does exist a non-extreme Nash equilibrium $(x_1^{*}, x_2^{*})$, $x_1^{*}+x_2^{*} \leq \frac{m_1+m_2}{2}$. and equals
to $\frac{m_1+m_2}{2}$ if and only if $m_{1}=m_{2}=\frac{m}{2}$.


\section{Uniqueness of Nash Equilibrium}
\label{uniqueness_section}

We have already studied three conditions illustrated in Corollary \ref{add_a_corollary}. Condition (a) corresponds to a unique extreme Nash equilibrium, so does condition (b). Condition (c) also corresponds to at most one non-extreme Nash equilibrium. It is natural to wonder whether this game has a unique pure Nash equilibrium. The following theorem guarantees the uniqueness.

\begin{theorem}
\label{uniqueness}
Every two-player miner’s dilemma game with betrayal assumption admits a unique pure Nash equilibrium.
\end{theorem}

\begin{proof}[proof of theorem \ref{uniqueness}]

Theorem \ref{existence} guarantees the existence of pure Nash equilibrium. To prove uniqueness, we need to prove at most one condition can be satisfied in Corollary \ref{add_a_corollary}. Since, $(0, 0)$ cannot be a Nash equilibrium by Lemma \ref{extreme0}, we should prove (a) and (c) cannot happen at the same time, as well as (b) and (c). W.L.O.G, we prove there cannot be two Nash equilibria that satisfy (a) and (c) respectively 
at the same time, and a symmetric analysis will hold for (b) and (c).

If $x_1^{*} = 0$ in one Nash equilibrium, we should have $g(t, p, m_1, m_2) \leq 0$ and $m_1 \leq \frac{(1-p)^2}{4}m_2 < m_2 $.

If there exists another non-extreme Nash equilibrium, the unique root $y^{*}$ for $f(y)=0$ should correspond to a feasible solution $(x_1^{*}, x_2^{*})$ within $(0, m_1)\times(0, m_2)$. Hence we must have $x_1^{*} > 0$. Solving this in Equation (\ref{x1,y}) we get $y^{*} < y_3$ or $y^{*} > y_4$, where
\begin{equation*}
\begin{aligned}
y_3 =& \frac{2m_1 + (1-p)m_2 + 2t - \sqrt{\Delta_2}}{2(1-p)},\\
y_4 =& \frac{2m_1 + (1-p)m_2 + 2t + \sqrt{\Delta_2}}{2(1-p)},
\end{aligned}
\end{equation*}
and $\Delta_2 = (2m_1 - (1-p)m_2)^2 + 4t^2 + 4t(2m_1 + (1-p)m_2) > 0$. 

$y_4 > \frac{2m_1 + (1-p)m_2 + 2t}{2(1-p)} > \frac{m_1 + m_2}{2}$, but we have proved $y^{*} \leq \frac{m_1 + m_2}{2}$, so only $y^{*} < y_3$ can happen. Referring to the proof of Lemma \ref{f_y}, if condition (c) corresponds to a solution, it means $f(y_3) > 0$. 

However, by substituting $y_3$ and solving $f(y_3) > 0$, we get exactly$$g(t, p, m_1, m_2) > 0, $$contradicting to $g(t, p, m_1, m_2) \leq 0$ stated before. Therefore, condition (a) and (c) contradict with each other.

\end{proof}


\citet{alkalay2019pure} points out another direction for establishing the uniqueness of pure Nash equilibrium is to leverage the result by \citet{rosen1965existence}, and show that the sufficient conditions they provide are satisfied in our game. Unfortunately, by our numerical experiments, the conditions required by Rosen Criterion are not satisfied in this game, even when $p=0$.

Note we have already fully depicted the unique Nash equilibrium in this game. Given $m_1, m_2, t$ and $p$, we determine the equilibrium point as follows. First we check the sign of function $g$ to determine whether it is an extreme case or not. If it is an extreme case, the Nash equilibrium follows Lemma \ref{extreme0}. If $m_1=m_2$, we get a symmetric Nash equilibrium by Lemma \ref{symmetric}. Otherwise, we solve Equation (\ref{f=0}) to get the root $y^{*}$ of function $f$, and then we can get a non-extreme equilibrium $(x_1^{*}, x_2^{*})$.

\section{Pure Price of Anarchy}
\label{ppoa_section}

Now we can prove our most significant result, which is about the pure price of anarchy of this game. A pure price of anarchy ($\mathrm{PPoA}$) is a measure of the ratio between the optimal social welfare and the worst social welfare in any pure Nash equilibrium. Following the analysis in \citet{alkalay2019pure}, we define the social welfare to be effective mining power of these two pools. The optimal mining power is $m_1 + m_2$ when no one attacks and the effective mining power is $m_1+m_2-(1-p)(x_{1}^{*}+x_{2}^{*})$ in a pure Nash equilibrium. Note by Theorem \ref{uniqueness}, we know the pure Nash equilibrium is unique, so given $m_1, m_2, t$ and $p$, the pure price of anarchy is

$$\mathrm{PPoA}=\frac{m_{1}+m_{2}}{m_1+m_2-(1-p)(x_{1}^{*}+x_{2}^{*})}.$$


The following theorem gives the tight bound of $\mathrm{PPoA}$.

\begin{theorem}
\label{main}
In every two-player miner’s dilemma game with betrayal assumption, the pure price of anarchy is at most $2$, and equal to $2$ if and only if $m_{1}=m_{2}=\frac{m}{2}$ and $p=0$. The tight bound of $\mathrm{PPoA}$ is $(1, 2]$.
\end{theorem}

\begin{proof}[proof of theorem \ref{main}]


By Lemma \ref{extreme0}, \ref{symmetric}, \ref{binary} and \ref{f_y}, we know $x_1^{*} + x_2^{*} = y^{*} \leq \frac{m_1+m_2}{2}$, so that $\mathrm{PPoA}\leq \frac{2}{1+p}$. The equality holds if and only if $x_1^{*} + x_2^{*} = y^{*} \leq \frac{m_1+m_2}{2}$, i.e., $m_{1}=m_{2}=\frac{m}{2}$. Then for $p\in [0, 1)$, $\mathrm{PPoA}\leq 2$, and equal to $2$ if and only if $p=0$.

The lower bound of $\mathrm{PPoA}$ is trivial. Since $(0, 0)$ cannot be a Nash equilibrium by Lemma \ref{extreme0}, $\mathrm{PPoA} > 1$ comes naturally. And taking the symmetric case as an example, by Lemma \ref{symmetric}, we have

\begin{equation*}
   \mathrm{PPoA} = \sqrt{\frac{(1+p)^2}{16(k-2)^2} + \frac{4(3-p)}{16(k-2)} + \frac{1}{4}} + \frac{1}{2} - \frac{p+1}{4(k-2)} \stackrel{k \to +\infty}{\longrightarrow} 1
\end{equation*}

\end{proof}

For every two-player miner’s dilemma game without betrayal assumption, \citet{alkalay2019pure} proved $\mathrm{PPoA} < 3$ in general case and gave the conjecture that pure Nash equilibrium is unique and $\mathrm{PPoA} \le 2$. We not only show the conjecture is true, but also consider the conjecture in a more general model, the game with betrayal assumption. Theorem \ref{uniqueness} proves the uniqueness and Theorem \ref{main} proves $\mathrm{PPoA} \le 2$, which is a tight upper bound.

If we keep $p$ in the upper bound, $\mathrm{PPoA} \le \frac{2}{1+p}$. While the upper bound of $\mathrm{PPoA}$ decreases in $p$, it is an interesting discovery that betrayal has no effect on the upper bound of the sum of infiltrating power, i.e., $x_{1}^{*}+x_{2}^{*} \le \frac{m_1+m_2}{2}$ holds tightly regardless of the value of $p$. In other words, Allowing miners to betray has no substantial impact on the macro strategy of the mining pool, but can reduce the systemic loss.

\section{Experiments}
\label{experiment_section}

Since the dimension of strategy profile space increases quadratically w.r.t. the number of players, the problem with more than two pools is much more complicated to analyze theoretically. Having thoroughly studied the case of two players, we focus on the $\mathrm{PPoA}$ of N-player miner's dilemma game and set $N=3$ in our experiments. Let $m_0$, $m_1$, $m_2$ denote the mining power of three pools, $t$ other mining power in the system, $p$ the betrayal parameter. $m_0, m_1, m_2, t$ are sampled as integers while the values of $p$ are equally spaced floating points in $[0, 1)$.


Let $x_{i,j}$ be the amount of mining power pool $i$ uses to attack pool $j$, satisfying $\sum_{j} x_{i,j} \leq m_i$, for $\forall i \in \{0,1,2\}$. For each setting, we iteratively calculate the best responses starting from honest mining, i.e., $x_{i,j}^0=0, \forall i,j \in \{0, 1, 2\}$. Specifically, in step $k$, pool $i$ chooses strategies $\{x^k_{i,j}: i\not=j \}$ that maximizes its reward based on others' strategies in step $k-1$: $\{x_{i^\prime,j}^{k-1}:i^\prime \neq j, i^\prime \neq i\}$. We assume the iteration converges to a Nash 
equilibrium when $\sum_{i, j: i\not=j}|x_{i,j}^{k} - x_{i,j}^{k-1}| \leq 2^{-18}$. The optimization process uses \textit{scipy} package in python.

\begin{figure}[!htbp]
\centering
\subfigure[]{
\label{fig_1}
\includegraphics[width=0.45\linewidth]{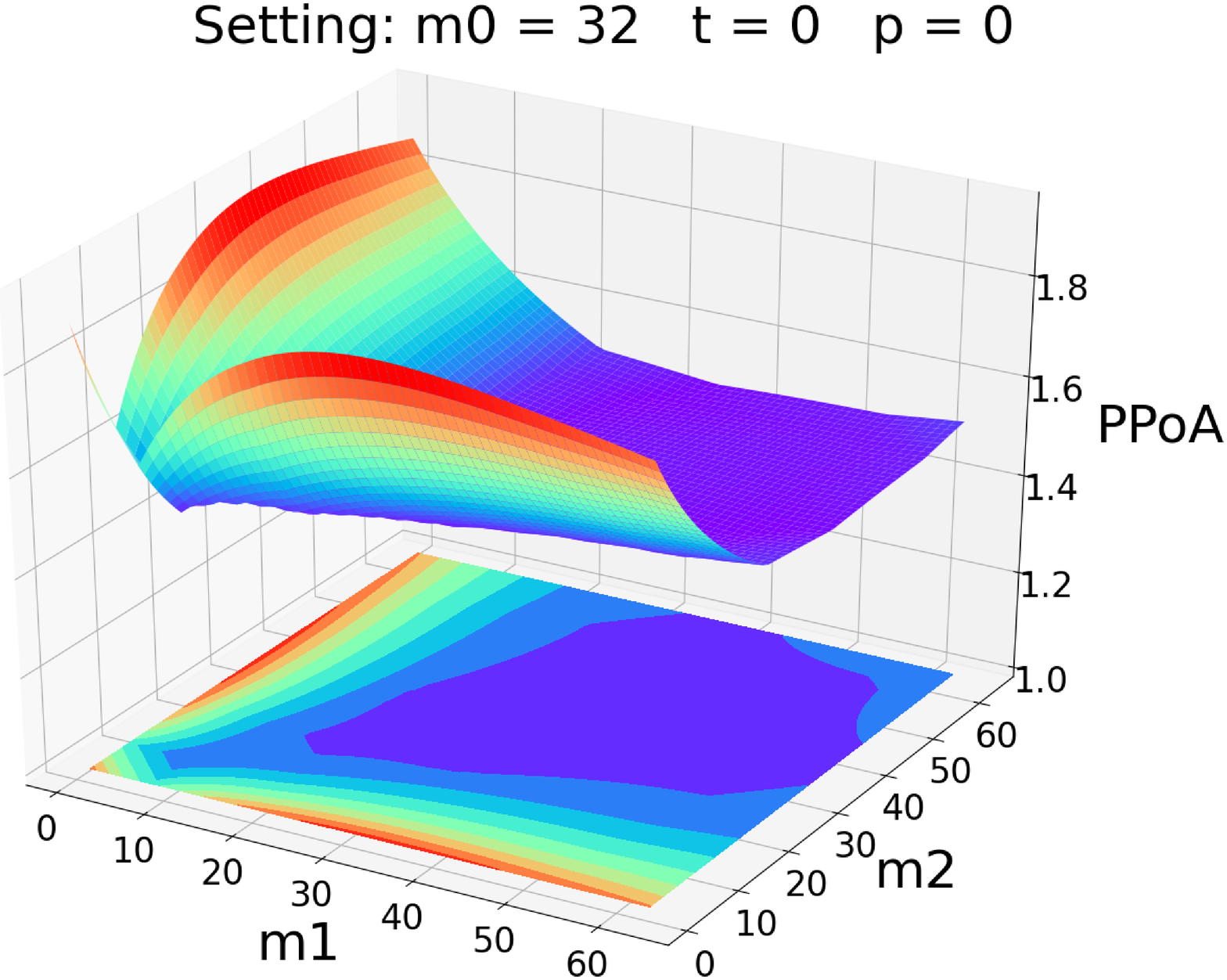}
}
\subfigure[]{
\label{fig_2}
\includegraphics[width=0.45\linewidth]{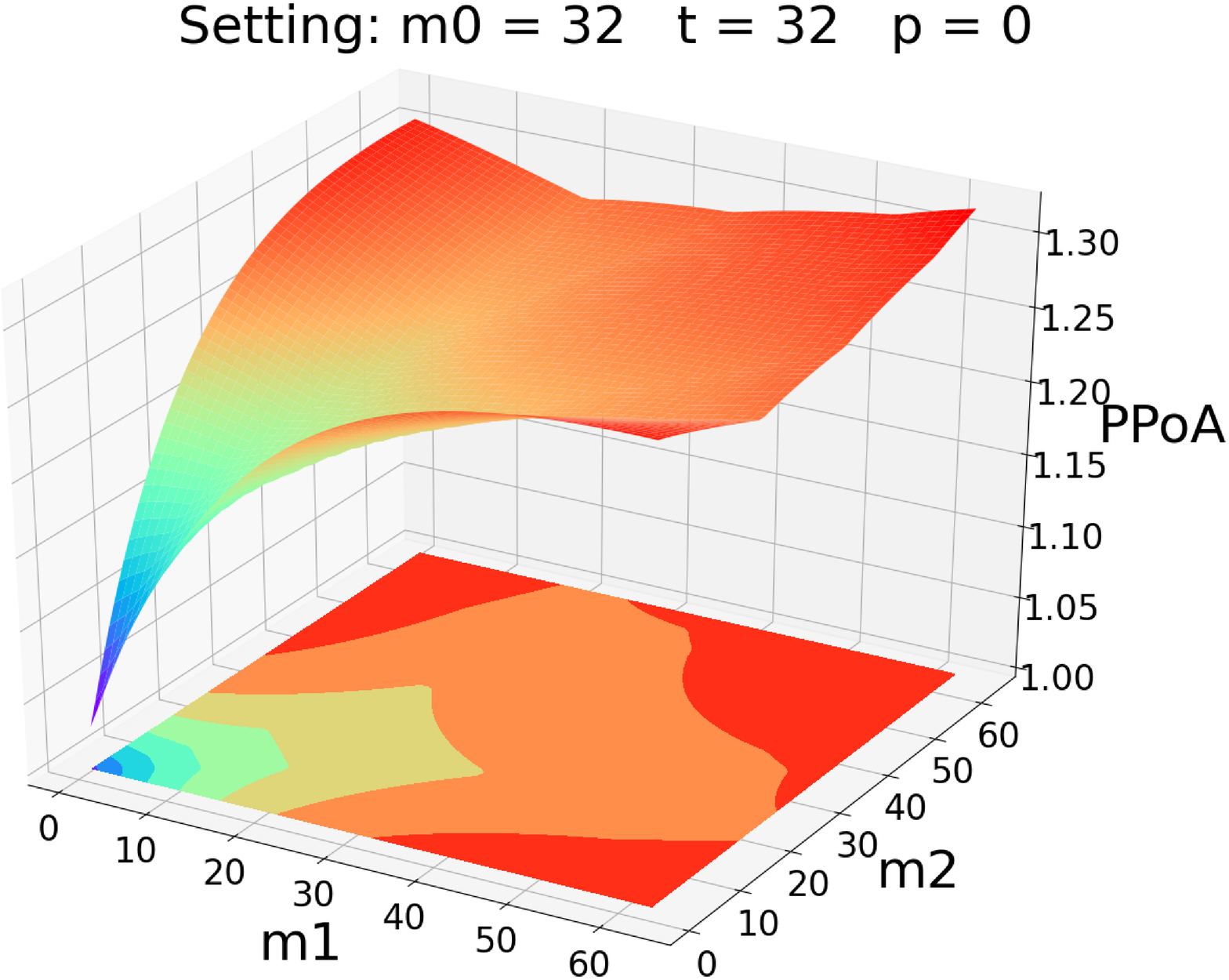}
}
\subfigure[]{
\label{fig_3}
\includegraphics[width=0.45\linewidth]{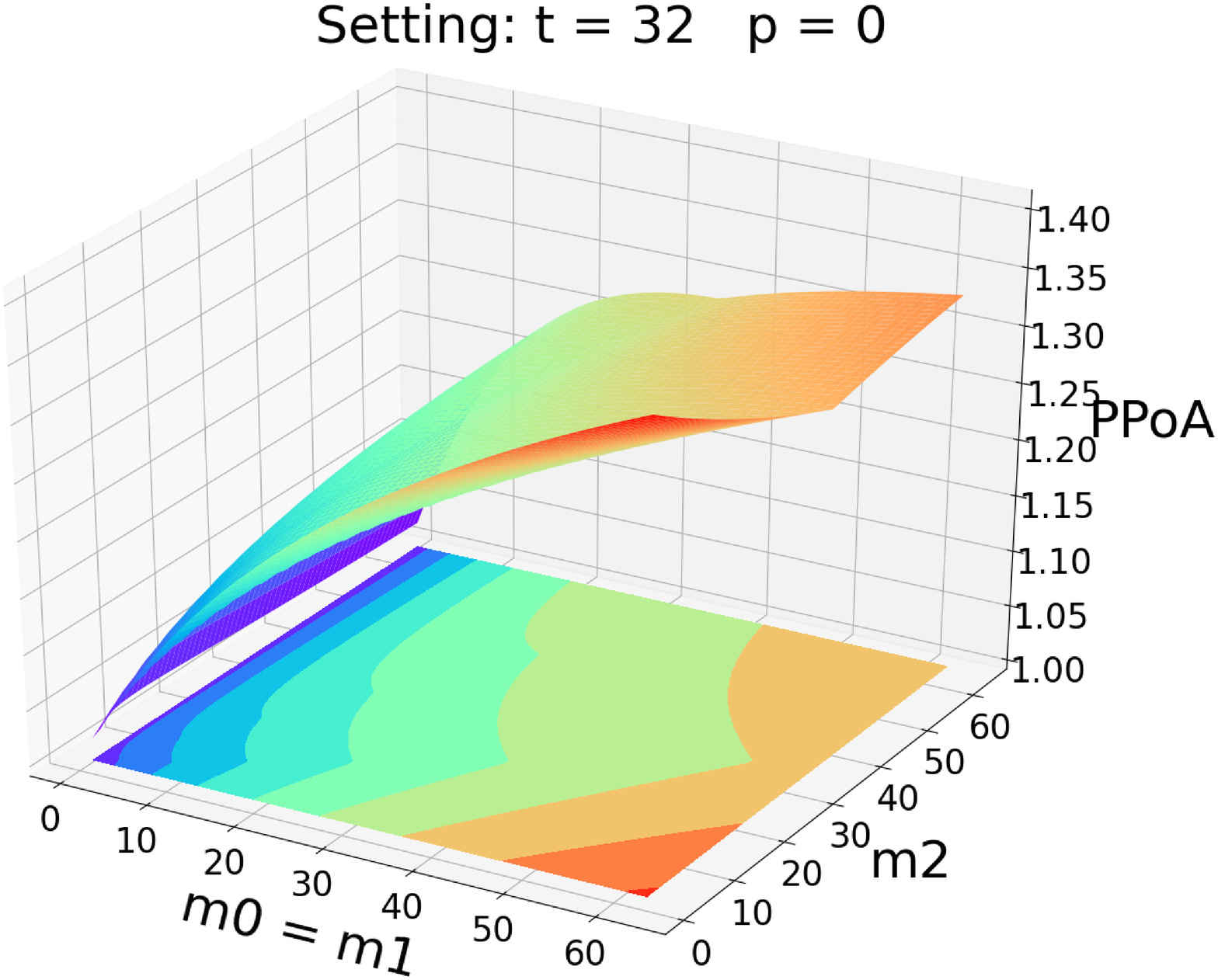}
}
\subfigure[]{
\label{fig_4}
\includegraphics[width=0.45\linewidth]{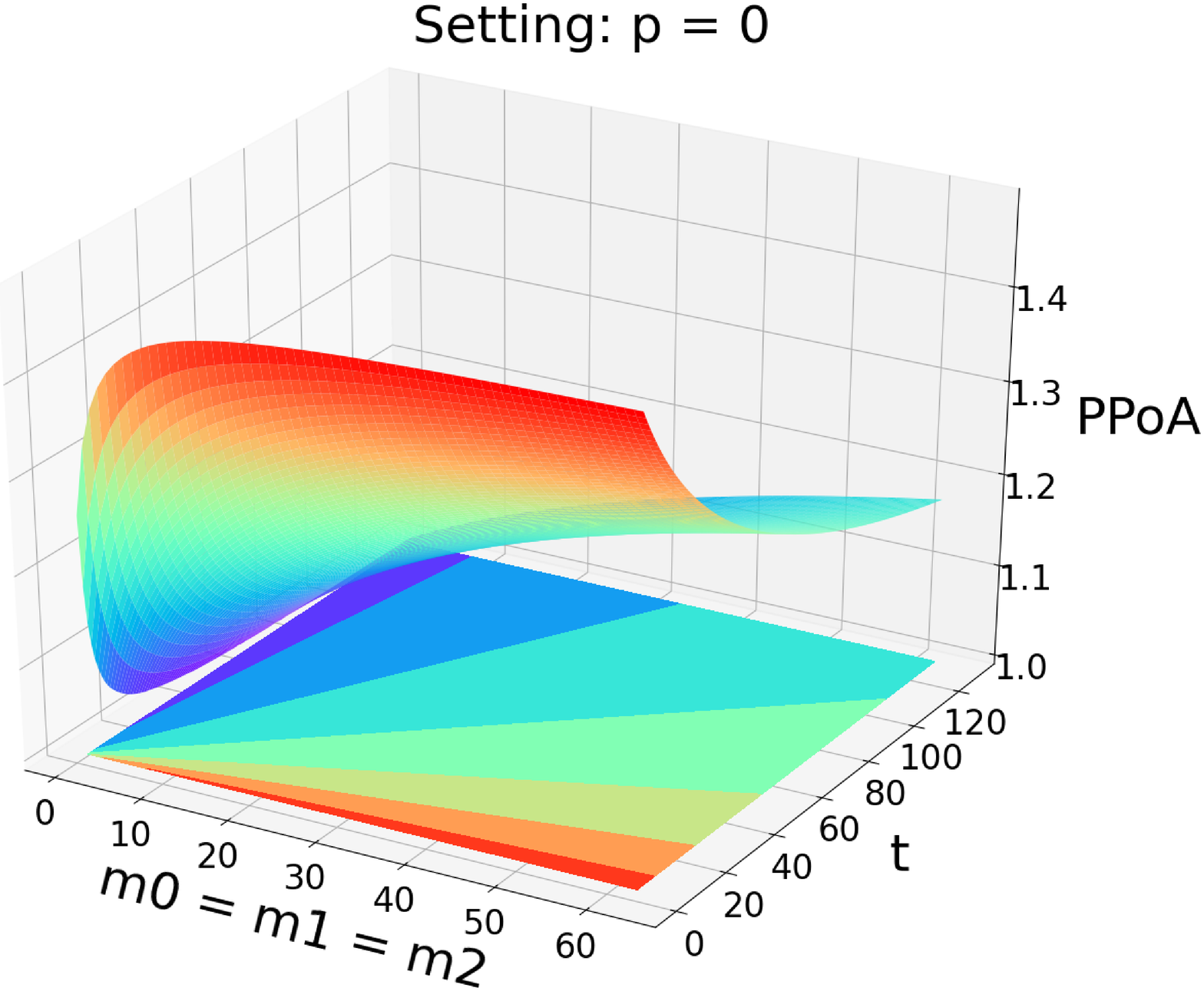}
}
\subfigure[]{
\label{fig_5}
\includegraphics[width=0.45\linewidth]{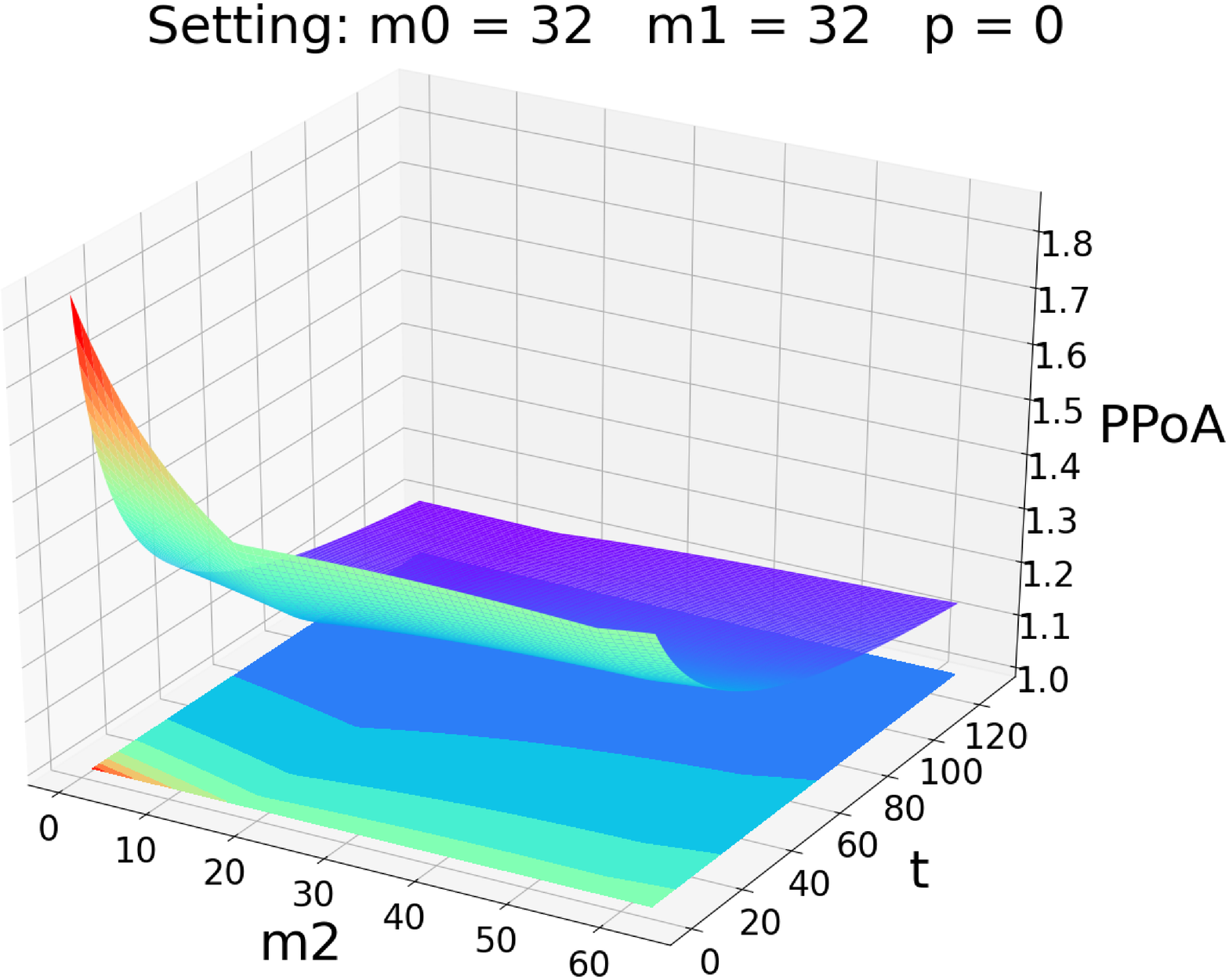}
}
\subfigure[]{
\label{fig_6}
\includegraphics[width=0.45\linewidth]{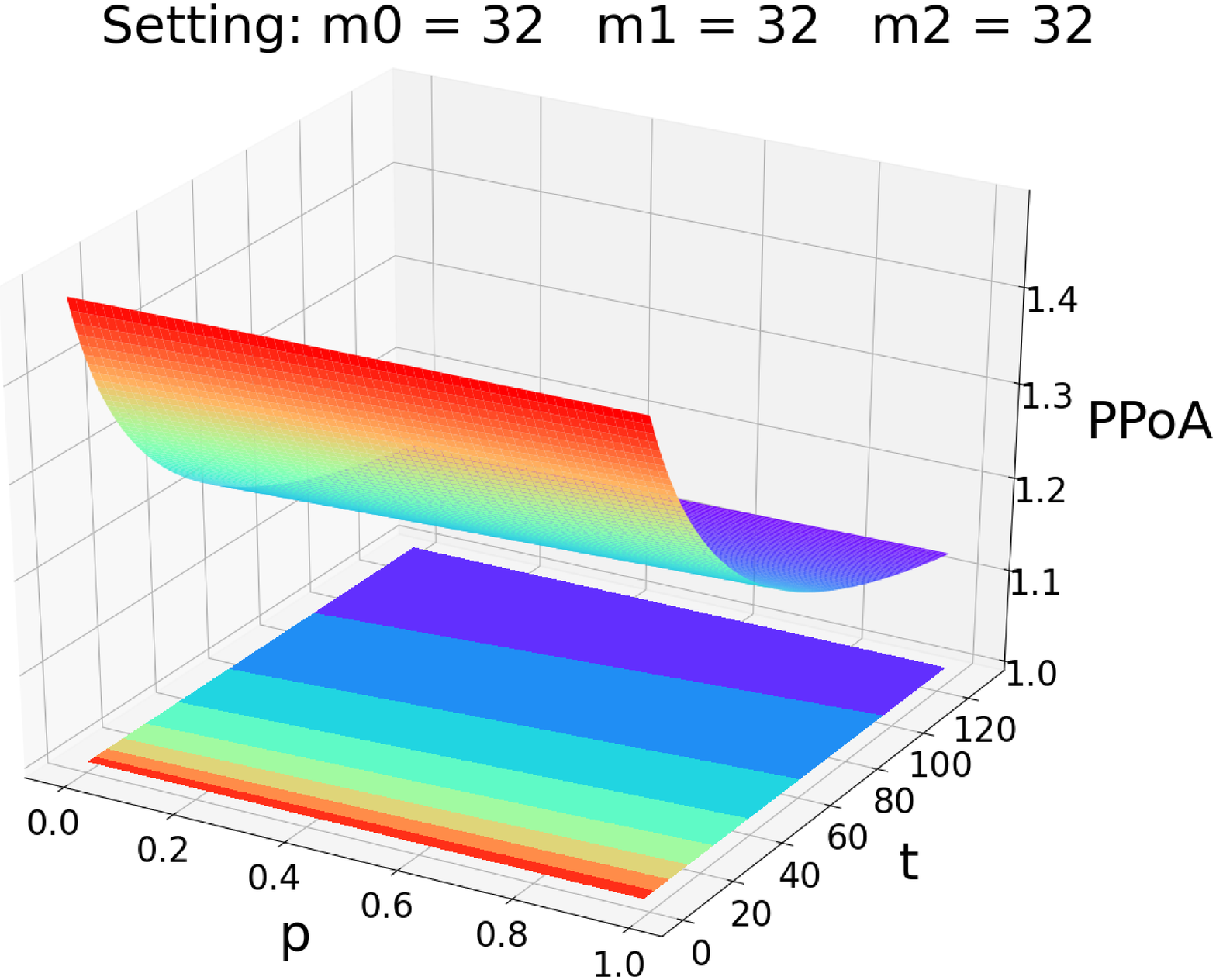}
}
\caption{Pure price of anarchy in 3-player miner’s dilemma game}
\label{fig}
\end{figure}

There are five parameters varying in this game, so we fix part of them to see how $\mathrm{PPoA}$ changes with other parameters, in the hope of catching its inner structure. In Figure \ref{fig_1} and \ref{fig_2}, we set $m_0=2^5, p=0$ and fix $t$ to a constant value ($t=0$ in \ref{fig_1}, $t=2^5$ in \ref{fig_2}). In Figure \ref{fig_3}, we set $m_0=m_1$ and let them change together. Then, we consider the symmetric case when $m_0 = m_1 = m_2$ and how $\mathrm{PPoA}$ changes w.r.t. $m_i$ and $t$, shown in Figure \ref{fig_4}.
In Figure \ref{fig_5}, we fix $m_0=m_1=2^5$. And finally we fix $m_0=m_1=m_2 = 2^5$ to see the effect of $p$ on $\mathrm{PPoA}$, as shown in Figure \ref{fig_6}.


Figure \ref{fig_4},\ref{fig_5}, \ref{fig_6} imply that $\mathrm{PPoA}$ will become smaller as t increases. That is to say, when this game involves a higher percentage of the mining power in the system, pools tend to use more mining power to perform attacks and the loss will become larger. It is also interesting to see with different $t$, Figure \ref{fig_1} and \ref{fig_2} show different trends of $\mathrm{PPoA}$ as $m_1$ and $m_2$ change. 

We can see in most instances, $\mathrm{PPoA} \leq \frac{3}{2}$, which means $\sum_{i, j} x_{i, j} \leq \frac{m_1+m_2+m_3}{3}$. $\mathrm{PPoA}$ will exceed $\frac{3}{2}$ only if the mining power of one pool approaches zero. In fact, 3-player miner's dilemma game degenerates to 2-player game in this case. Although our experiment only covers a few special settings, it can be intuitively judged that decentralization can help reduce the mining power wasted in the game. We conjecture that the following stronger result should hold.

\begin{conjecture}
\label{conj}
Every $N$-player ($N \geq 2$) miner’s dilemma game with betrayal assumption admits a unique pure Nash equilibrium, and the tight bound of pure price of anarchy is $(1, 2]$.
\end{conjecture}

\section{Conclusion}

We explore the pool block withholding attacks and make an in-depth game theoretic analysis. Compared to previous models for miner's dilemma, our model is more general with the introduction of betrayal assumption. Focusing on the pure Nash equilibrium and the pure price of anarchy, we prove that every two-player miner’s dilemma game with betrayal assumption admits a unique pure Nash equilibrium and the tight bound of $\mathrm{PPoA}$ is $(1, 2]$. We also give the explicit expression of Nash equilibrium in special cases.

Proving Conjecture \ref{conj}, i.e. establishing uniqueness of pure Nash equilibrium and proving an upper bound of 2 on $\mathrm{PPoA}$ in N-pool case ($N\geq 3$), is one of our future works. The analysis of mixed Nash equilibria is also an interesting research direction. 






\newpage

\bibliographystyle{ACM-Reference-Format} 
\bibliography{ref}






\end{document}